\documentclass[12pt]{article}  
\usepackage[left=3cm, right=3cm, top=3cm, bottom=3cm]{geometry}
\usepackage{amsmath,amsthm,amssymb,natbib,enumerate,xcolor}
\usepackage[hyphens]{url}
\usepackage{hyperref}
\hypersetup{
    colorlinks,
    linkcolor={red!50!black},
    citecolor={black},
    urlcolor={black}
}
\newtheorem{thm}{Theorem}
\newtheorem{prop}[thm]{Proposition} 
\newtheorem{lemma}[thm]{Lemma}

\theoremstyle{definition}
\newtheorem{defi}{Definition}
\theoremstyle{remark}

\newcommand*{\set}[1]{\left\{#1\right\}}
\newcommand*{\good}{G}
\newcommand*{\bad}{B}
\newcommand*{\x}{X}
\newcommand*{\y}{Y}
\newcommand*{\w}{w}

\title{Competitive pricing despite search costs if lower price signals quality}
\author{Sander Heinsalu\thanks{Research School of Economics, Australian National University. HW Arndt Building, 25a Kingsley St, Acton ACT 2601, Australia.
Email: sander.heinsalu@anu.edu.au, 
website: \url{http://sanderheinsalu.com/}
\newline
The author thanks Larry Samuelson, George Mailath, Gabriel Carroll, Gabriele Gratton, Jack Stecher, David Frankel, Simon Anderson, Idione Meneghel, Vera te Velde and Andrew McLennan for comments and suggestions.
}}
\date{\today}

\begin{document}
\maketitle

\begin{abstract}
%
%
I show that firms price almost competitively and consumers can infer product quality from prices in markets where firms differ in quality and production cost, and learning prices is costly. Bankruptcy risk or regulation links higher quality to lower cost. If high-quality firms have lower cost, then they can signal quality by cutting prices. Then the low-quality firms must cut prices to retain customers. This price-cutting race to the bottom ends in a separating equilibrium in which the low-quality firms charge their competitive price and the high-quality firms charge slightly less. 

Keywords: Price signalling, Diamond paradox, price dispersion, incomplete information, price war. 
	
JEL classification: D82, C72, D41. 
\end{abstract}

In many markets, looking up prices is costly for consumers. Even in online markets, the few moments it takes to check a (second) website constitutes a positive cost. 
Firms usually have private information about their cost and quality. This paper shows that with negatively related private cost and quality, firms set a price that is perfectly competitive or close to it. 
The negative association of cost and quality can arise from several causes. 
A more skilled tradesperson or a firm with better equipment can provide higher-quality service with less time and effort, thus at a lower cost. Examples are tire change and rotation using a specialised machine versus `by hand', ironing a shirt using a dummy (e.g.\ Siemens Dressman) extruding hot air, measuring distance with a laser rangefinder instead of tape. 
Economies of scale imply a lower marginal cost for larger producers, and learning by doing improves their quality. For example, a larger insurer is less risky (better for policyholders) and has lower overhead costs per policy.\footnote{Warren Buffett's 2014 letter to shareholders (\url{http://www.berkshirehathaway.com/letters/letters.html}) p.~11 describes the higher underwriting profit, lower cost and smaller risk of default of Berkshire Hathaway insurance businesses relative to competitors.
} 
Amazon has greater variety, faster delivery and a lower cost per package delivered than smaller online sellers. 

Regulation can cause a low-quality firm to have a higher cost. 
If a low-quality firm's product is more likely to be faulty and a regulator punishes firms for faulty products, then a low-quality firm has a greater cost (of production plus expected punishment) per product sold. Consumer lawsuits over bad quality have a similar effect. 
Regulation can also turn a low cost into an incentive to improve quality. A low-cost firm optimally prices lower than a high-cost one. If there is no quality difference, then demand is greater for a low-cost firm. If a regulator checks firms with a larger market share more often and punishes bad quality, then larger producers (those with lower cost) have a greater incentive to improve quality. Higher quality further increases demand for a low-cost firm. 

Optimal allocation of managerial talent (or some other resource) between cost reduction and quality improvement also links lower cost to higher quality. Improving quality is subject to moral hazard, because consumers pay based on the quality they expect, not the quality that the firm chooses. Cost reduction benefits the firm directly, so it is optimal to reduce cost maximally before improving quality. If managerial talent differs between firms, then those with high talent reduce their cost to a minimum and then might as well improve quality. The low-talent firms stay above the minimal cost and don't improve quality. 

In the markets studied in this paper, there are at least two firms, each of which draws an independent type, either \emph{good} or \emph{bad}. The good type has lower marginal cost and higher quality than the bad. Each firm knows its own type. The consumers and the other firms only have a common prior over a firm's type. First the firms simultaneously set prices. Second, each consumer observes the price of one firm for free and chooses either to buy from this firm, leave the market or pay a small cost to learn the price of another firm. Finally, each consumer who learned chooses either to buy from one of the firms whose price he knows or leave the market. The consumers have a distribution of valuations. A higher-valuation consumer values high quality relatively more. 
Consumers update their beliefs about the type of a firm whose price they see using Bayes' rule whenever possible. 
The equilibrium concept is perfect Bayesian equilibrium (PBE). A unique equilibrium remains after 
refining with the Intuitive Criterion. 

In equilibrium, prices are close to competitive 
due to a \emph{race to the bottom} that consists of two forces. One is downward price signalling, i.e.\ the high-quality firm reduces price to distinguish itself from the low-quality firm and attract greater demand. The second force is that a low-quality firm cuts price to deter consumers from leaving. The low-quality firms are in Bertrand competition over the consumers who learn more than one price, which all consumers do when faced with a price indicative of a low-quality firm. 

The race to the bottom ends when the high-quality firm prices at the marginal cost of the low-quality firm. The low-quality firm's price is its marginal cost plus the minimal monetary unit---the same as under complete-information Bertrand competition between two low-quality firms when the consumers have zero learning cost. 
Bertrand competition between two known high-quality firms leads to a lower price than under incomplete information, but between a known high-quality and low-quality firm to a higher price. 
The close-to-competitive pricing contrasts with the paradoxical result of \cite{diamond1971} that without uncertainty about the costs and qualities of the firms, the unique equilibrium features monopoly pricing and no consumer learning. 


Competition with privately known cost and quality also contrasts with monopoly under private information, and with competition when quality is observed together with the price. 
In monopoly, the good type still signals its quality by reducing its price to a level less profitable for the bad type than the bad type's monopoly price. However, the bad type has no incentive to cut price below its monopoly price. 

In competition when paying a learning cost leads to observing both price and quality, the low-quality firms are still in a Bertrand-like situation and compete to a low price. However, a high-quality firm has no incentive to reduce price to signal, because customers stay at a high-quality firm even at a high price. 
The combination of competition and unobservable firm characteristics is thus necessary for competitive pricing, as well as sufficient. 

Empirical evidence for the model can be found in the car industry, 
namely that high-quality cars have a lower price and production cost. 
According to \cite{vasilash1997}, the assembly cost of more reliable cars is smaller, controlling for vehicle category, e.g.\ subcompact, compact, etc. 
The rank correlation between the price\footnote{The price used is the average of the average low and high price paid from the US News \& World Report Car Ranking and Advice \url{https://cars.usnews.com/cars-trucks/
}. Author's calculations.}
and the CarMD  index\footnote{\url{https://www.carmd.com/wp/vehicle-health-index-introduction/2015-carmd-manufacturer-vehicle-rankings/}} 
of repair incidents (higher index means more breakdowns) is positive, but statistically insignificant in the sample of 40 cars that belong to both the top 288 new cars by number sold in the US in 2015\footnote{\url{http://www.goodcarbadcar.net/2015/07/usa-20-best-selling-cars-june-2015-sales-figures.html}} 
and the 100 most reliable in 2015 according to CarMD. 
The average cost per repair is also larger for more expensive cars according to CarMD, but this is less surprising, because parts for more expensive vehicles cost more. 
The cheapest cars to maintain according to  YourMechanic\footnote{\url{https://www.yourmechanic.com/article/the-most-and-least-expensive-cars-to-maintain-by-maddy-martin}} are those of East Asian manufacturers, with Toyota leading. These are also the cheapest to buy according to the US News \& World Report. 

Given the above, it is not surprising that the per-car profits are higher for Toyota than for Detroit's Big 3 automakers \citep{wayland2015}. 
Competition has resulted in approximately zero profit for the higher-cost manufacturers. For example, the US government had to bail out Detroit's Big 3 carmakers during the 2008 financial crisis. 
Price close to marginal cost (profit close to zero) for low-quality producers is consistent with the model. The almost zero correlation of price and quality also corresponds to the model, because the prices of producers of different quality are close to each other in equilibrium.

\subsubsection*{Literature} 
The foremost paper on costly learning of prices is \cite{diamond1971}, where competing firms set the monopoly price. A monopoly price or above is also found in \cite{diamond1987,axell1977,reinganum1979,klemperer1987} and \cite{garcia+2017}. 
A number of solutions to the Diamond paradox have been proposed. With a positive fraction of consumers having zero learning cost, as in \cite{butters1977,stahl1996,klemperer1987} and \cite{benabou1993}, firms put a positive probability on the competitive price. 
A similar idea to zero learning cost is that consumers observe multiple prices with positive probability, for example by seeing price advertisements \citep{salop+stiglitz1977,burdett+judd1983,robert+stahl1993}. 
If consumers have private taste shocks, then that generates search and below-monopoly pricing \citep{wolinsky1986,anderson+renault1999,zhou2014}. 
Prices below the monopoly level also occur with repeat purchases, as in \cite{salop+stiglitz1982,bagwell+ramey1992}. 

The current paper does not rely on zero learning cost, multiple free price observations, taste shocks or repeat purchases. To the author's knowledge, this work is the first to combine signalling and consumer search costs. 
The informative price difference between firm types endogenously gives consumers the incentive to learn, in contrast with the exogenous incentive created by taste shocks or a zero search cost. Multiple free price observations constitute exogenous learning, also differing from an endogenous motivation to learn. 
In the current work, the incentive for firm types to set different and low prices is endogenous, driven by consumer beliefs responding to the price. 
This differs from \cite{salop+stiglitz1982} where firms are indifferent between selling two units at a lower price or one unit at a higher, and these prices are determined by the exogenous willingness to pay of consumers. 
In \cite{bagwell+ramey1992}, the motivation for a low price is that in the infinitely repeated game, consumers start to boycott firms that raise price. This motivation is endogenous, but different from the current paper. 


Downward\footnote{As opposed to the upward price signalling (higher-quality firm sets a higher price) studied by the large literature following \cite{milgrom+roberts1986}.
} 
price signalling by a single firm has been studied in \cite{shieh1993}. A similar idea is in \cite{simester1995}, where multiproduct firms (whose prices for all products are positively correlated) signal by a low price on one product. 
In \cite{rhodes2015}, a multiproduct monopolist stocking more products (better for the consumers) charges lower prices. The result is similar to a higher-quality firm charging less, but the mechanism is different: adding a product attracts additional customers with relatively low valuations for the other products. When the average valuation of customers falls, the monopoly price falls. 

The receivers of the price signal are the consumers in this paper, which differs from limit pricing (as in \cite{milgrom+roberts1982b} and the literature following) where the receivers are potential entrants. 

The next section sets up the model. 
Section~\ref{sec:existence} constructs an equilibrium with near-competitive pricing in a market with consumer search costs. 
and shows that this equilibrium is the unique one 
that survives the Intuitive Criterion of \cite{cho+kreps1987}. The robustness of the results to relaxing various assumptions is discussed in Section~\ref{sec:robust}. 

\section{Price competition with costly learning of prices}
\label{sec:setup}

There are two firms indexed by $i\in\set{\x,\y}$, each with a type $\theta\in\set{\good,\bad}$ (good and bad, respectively). 
Each firm knows its own type, but not that of the other. Types are i.i.d.\ with $\Pr(\good)=\mu_0\in(0,1)$. There is a continuum of consumers of mass $1$ with types $v\in[0,\overline{v}]$ distributed according to the strictly positive continuous pdf $f_v$, with cdf $F_v$, independently of firm types. 
Firms and consumers know their own type, but only have a common prior belief over the types of others. 

The timeline of the game is as follows.
\begin{enumerate}
\setlength\itemsep{0.0pt}
\item Nature draws independent types for firms and consumers, and assigns half the consumers to one firm, half to the other, independently of types. Each player observes his own type, but not the types of the others.
\item Firms simultaneously set prices. 
\item Each consumer observes the price of his assigned firm and chooses either to buy from this firm, learn the price of the other firm, or leave the market.
\item Each consumer who chose to learn observes both firms' prices and chooses either to buy from his assigned firm, buy from the other firm, or leave the market. 
\end{enumerate}

A type $\good$ firm has marginal cost $c_{\good}$ normalised to $0$, and type $\bad$ has $c_{\bad}>0$. The quality of a type $\good$ firm is higher. 
Specifically, a type $v$ consumer values firm type $\bad$'s product at $v$ and $\good$'s product at $h(v)\geq v$, with $h'>1$, $h(\overline{v})<\infty$. To ensure that demand for $\bad$'s good is positive, assume $\overline{v}>c_{\bad}$. Consumers and firms are risk-neutral. Each consumer has unit demand. 

After the firms' cost and quality are determined, the firms simultaneously set prices $P_{\x},P_{\y}\in S_P:=\{0,m,2m,\ldots,Nm\}$, where $m>0$ is the smallest monetary unit.\footnote{
Using a discrete price grid avoids problems with equilibrium existence (explained in Section~\ref{sec:robust}), which are not the focus of this paper. 
An alternative to the grid is to restrict prices to $ \mathbb{R}_+\setminus (c_{\bad}-\rho,c_{\bad}+\rho)$ for some $\rho>0$. 
} 
Assume $c_{\bad}=km\geq h(0)-m$ for some $k\in\mathbb{N}$ (costs are measured in terms of the minimal monetary unit, and not all consumers buy at a price just below the bad type's cost). Assume $Nm\geq h(\overline{v})$. Prices above $Nm$ are unavailable w.l.o.g., because no consumer buys at any $P>h(\overline{v})$. 

For a set $S$, denote the set of probability distributions on $S$ by $\Delta S$. 
A behavioural strategy of firm $i$ is $\sigma_i:\set{\good,\bad}\rightarrow \Delta S_P$, so $\sigma_i(\theta)(P)$ is the probability that type $\theta$ of firm $i$ puts on price $P$.

A consumer sees the price that his assigned firm sets and can learn the price of the other firm at cost $c_{\ell}>0$. 
Define $c_{\bad}^{+}:=c_{\bad}+m$. Assume that $c_{\ell}\leq \mu_0 (h(c_{\bad}^{+})-c_{\bad})$, i.e.\ the learning cost is small relative to the prior probability of the good type firm and the valuation difference between consumer type $c_{\bad}^{+}$ for a good type firm and consumer type $c_{\bad}$ for a bad type firm. 
Assume that $m<\min\{c_{\ell},\;\frac{\mu_0 c_{\bad}}{1-\mu_0},\;\overline{v}-c_{\bad}\}$, i.e.\ the minimal monetary unit is small relative to the costs, the prior, and the maximal valuation for the bad type. 
The cost difference $c_{\bad}-0$ between the types, as well as the quality difference $h(0)-0$ may be small, provided the learning cost and minimal monetary unit are even smaller. 

After seeing the price of his assigned firm, a consumer decides whether to buy from this firm (denoted $b$), learn the other firm's price ($\ell$) or not buy at all ($n$). Upon learning the price of the other firm, the consumer decides whether to buy from firm $\x$ (denoted $b_{\x}$), firm $\y$ ($b_{\y}$) or not buy at all ($n_{\ell}$). 
A consumer's behavioural strategy consists of $\sigma_1:[0,\overline{v}]\times S_P\rightarrow\Delta\set{b,n,\ell}$ and $\sigma_2:[0,\overline{v}]\times S_P^2\rightarrow \Delta\set{b_{\x},b_{\y},n_{\ell}}$, so that e.g.\  $\sigma_2(v,P_i,P_j)(b_{j})$ is the probability that a consumer type $v$ initially at firm $i$ buys from $j\neq i$ after learning $P_j$. 

A type $\theta$ firm's \emph{ex post} payoff if mass $D$ of consumers buy from it at price $P$ is $(P-c_{\theta})D$. 
Assume that the full-information monopoly profit function $P[1-F_v(h^{-1}(P))]$ of firm type $\good$ strictly increases in $P$ on $[0,c_{\bad}+m]$, so that the full-information monopoly price $P_{\good}^{m}$ of $\good$ is strictly above $c_{\bad}$ (this is relaxed in Section~\ref{sec:robust}). 

A consumer's posterior belief about firm $i$ after observing its price $P_i$ and expecting the firm to choose strategy $\sigma_i^*$ is  
\begin{align}
\label{mu}
\mu_i(P_i) :=\frac{\mu_0\sigma_i^*(\good)(P_i)}{\mu_0\sigma_i^*(\good)(P_i) +(1-\mu_0)\sigma_i^*(\bad)(P_i)} 
\end{align}
whenever $\mu_0\sigma_i^*(\good)(P_i) +(1-\mu_0)\sigma_i^*(\bad)(P_i)>0$, and arbitrary otherwise. 
The gain from trade that consumer type $v$ expects from buying from firm $i$ at price $P$ is denoted $\w(v,i,P):=\mu_i(P)h(v)+(1-\mu_i(P))v-P$. 

The solution concept used is perfect Bayesian equilibrium (PBE), hereafter simply called equilibrium. Later, a unique equilibrium 
is selected using the Intuitive Criterion of \cite{cho+kreps1987}. 
\begin{defi}
\label{def:mix}
An equilibrium consists of $\sigma_{\x}^*,\sigma_{\y}^*,\sigma_1^*,\sigma_2^*$ and $\mu_{\x},\mu_{\y}:S_{P}\rightarrow [0,1]$ satisfying the following for $\theta\in\set{\good,\bad}$, $v\in[0,\overline{v}]$, $i,j\in\set{\x,\y}$, $i\neq j$: 
\begin{enumerate}[(a)]
\item if $\w(v,i,P_{i})\geq \max\set{0,\;\w(v,j,P_{j})}$, then $\sigma_2^*(v,P_{i},P_{j})(b_i)=1$, and if in addition $\w(v,i,P_{i})> \w(v,j,P_{j})$, then $\sigma_2^*(v,P_{j},P_{i})(b_i)=1$, 
\item if $\max\set{\w(v,i,P_{i}),\; \w(v,j,P_{j})}<0$, then $\sigma_2^*(v,P_{i},P_{j})(n_{\ell})=1$, 
\item if $\w(v,i,P_{i})> \max\{0,\; \sum_{P_{j}\in S_P}\max\{\w(v,i,P_{i}),\;\w(v,j,P_{j})\}[\mu_0\sigma_j^*(\good)(P_{j}) +(1-\mu_0)\sigma_j^*(\bad)(P_{j})] -c_{\ell}\}$, 
then $\sigma_1^*(v,P_{i})(b)=1$,
\item if $\w(v,i,P_{i})\leq \sum_{P_{j}\in S_P}\max\{0,\;\w(v,i,P_{i}),\;\w(v,j,P_{j})\}[\mu_0\sigma_j^*(\good)(P_{j}) +(1-\mu_0)\sigma_j^*(\bad)(P_{j})] -c_{\ell}\geq 0$, 
then $\sigma_1^*(v,P_{i})(\ell)=1$,
\item if $\max\{\w(v,i,P_{i}),\; \sum_{P_{j}\in S_P}\max\{0,\;\w(v,j,P_{j})\}[\mu_0\sigma_j^*(\good)(P_{j}) +(1-\mu_0)\sigma_j^*(\bad)(P_{j})] -c_{\ell}\}< 0$, 
then $\sigma_1^*(v,P_{i})(n)=1$, 
\item if $\sigma_i^*(\theta)(P_{i})>0$, then $P_{i}\in\arg\max_{P\in S_P} (P-c_{\theta})D_i(P)$, where 
\begin{align}
\label{demand}
&D_i(P) :=\frac{1}{2}\int_0^{\overline{v}} \sum_{P_{j}\in S_P} \{\sigma_1^*(v,P)(b) +\sigma_1^*(v,P)(\ell)\sigma_2^*(v,P,P_{j})(b_i) \\&\notag+\sigma_1^*(v,P_{j})(\ell)\sigma_2^*(v,P_{j},P)(b_i)\}[\mu_0\sigma_j^*(\good)(P_j) +(1-\mu_0)\sigma_j^*(\bad)(P_j)] dF_v(v),
\end{align}
\item if $\sigma_i^*(\good)(P)>0$ or $\sigma_i^*(\bad)(P)>0$, then $\mu_i(P)$ is derived from~(\ref{mu}). 
\end{enumerate}
\end{defi}
The equilibrium profit of type $\theta$ of firm $i$ is denoted $\pi_{i\theta}^*$; it equals $(P-c_{\theta})D(P)$ for any $P$ s.t.\ $\sigma_i^*(\theta)(P)>0$. 

Some tie-breaking rules are built into the equilibrium definition, e.g.\ a consumer indifferent between $b$ and $\ell$ chooses $\ell$. The results remain substantially the same if the tie-breaking rules are modified, as discussed in Section~\ref{sec:robust}.
The next section constructively proves equilibrium existence by guessing and verifying.

\section{Equilibrium}
\label{sec:existence}

This section constructs an equilibrium in which consumers put probability one on a firm being the good type if the price is below the bad type's cost, otherwise probability one on the bad type. The good type firm sets a price equal to the bad type's cost. The bad type's price is its cost plus $m$. A consumer initially facing a price less than the bad type's cost either buys (when his valuation for the good type is above the price) or leaves the market. A consumer who initially sees a price strictly above the bad type's cost learns (when his expected valuation for the other firm is above $c_{\ell}$) or leaves the market. After learning, all consumers buy from the lower-priced firm or leave the market. 

The formal definition of the \textbf{guessed equilibrium} is the following: 
\begin{enumerate} 
\item Beliefs: $P\leq c_{\bad}\Rightarrow\mu_i(P)=1$ and $P> c_{\bad}\Rightarrow \mu_i(P)=0$  for $i\in\set{\x,\y}$.
\item Each firm's type $\good$ sets price $c_{\bad}$ and type $\bad$ sets $c_{\bad}^{+}$. 
\item If $\mu_i(P)=1$, then $h(v)\geq P \Rightarrow \sigma_1^*(v,P)(b)=1$ and $h(v)< P\Rightarrow \sigma_1^*(v,P)(n)=1$. 
\item If $\mu_i(P)=0$ and $\mu_0 (h(v)-c_{\bad})+(1-\mu_0)(v-c_{\bad}^{+})-c_{\ell}\geq0$, then $\sigma_1^*(v,P)(\ell)=1$. 
\\If $\mu_i(P)=0$ and $\mu_0 (h(v)-c_{\bad})+(1-\mu_0)(v-c_{\bad}^{+})-c_{\ell}< 0$, then $\sigma_1^*(v,P)(n)=1$. 
\item If $\w(v,i,P_{i})\geq \max\set{0,\;\w(v,j,P_{j})}$, then $\sigma_2^*(v,P_{i},P_{j})(b_i)=1$, and if in addition $\w(v,i,P_{i})> \w(v,j,P_{j})$, then $\sigma_2^*(v,P_{j},P_{i})(b_i)=1$. 
If $\max\set{\w(v,i,P_{i}),\; \w(v,j,P_{j})}<0$, then $\sigma_2^*(v,P_{i},P_{j})(n_{\ell})=1$. 
\end{enumerate}
Part 1 of the guessed equilibrium includes Definition~1(g) and also specifies beliefs at off-path prices. Beliefs are consistent with Bayes' rule~(\ref{mu}). Part 2 specialises Definition~1(f) to the guessed equilibrium. Parts 3--5 are simply the rewriting of Definition~1(a)--(e). 
Appendix~\ref{sec:existenceproof} proves that no player can profitably deviate from the guessed equilibrium. 

The idea of the proof is as follows. Consumers are clearly best responding to their belief, which is consistent with firm strategies. The bad type does not price below $c_{\bad}$, because it is weakly dominated by $c_{\bad}^{+}$. If consumers at a bad type learn and the other firm is the good type, then all consumers leave the bad type. Otherwise, the two bad types are in Bertrand competition over the consumers who learn. So the bad types undercut each other until pricing at $c_{\bad}^{+}$. 
A good type does not increase price above $c_{\bad}$, because the resulting fall in belief reduces expected profit below that obtained at price $c_{\bad}$. The reason is twofold. If the other firm is the good type, then all consumers leave. If the other firm is the bad type, then no consumers are drawn away from that firm, which would happen at price $c_{\bad}$ or below. At prices less than $c_{\bad}$, the Diamond paradox reasoning applies to the good types: each can raise its price to $m$ above that of the other without losing demand. This is because the consumers' learning cost is above $m$, so a price $m$ higher than expected does not motivate them to learn and switch, unless the price increase changes their belief. 

The guessed equilibrium already partly resolves the Diamond paradox, because its outcome differs from monopoly pricing and no search. Prices in the guessed equilibrium are close to competitive. 
Type $\bad$ prices the same as under Bertrand competition between the $\bad$ types with zero search cost and complete information. Type $\good$ prices below $\bad$. 
For a stronger resolution of the Diamond paradox, 
subsequent results will show that the guessed equilibrium introduced above is the unique one 
that survives the Intuitive Criterion. 
Without refinement, belief threats support other equilibria. For example, for high enough $\mu_0$, both firms pool on $c_{\bad}^{+}$, justified by the belief $\mu_i(c_{\bad}^{+})=\mu_0$ and if $P\neq c_{\bad}^{+}$, then $\mu_i(P)=0$. 

The following lemma shows the monotonicity of equilibrium demand and prices. Given the ranking of the costs and qualities of the types, the results are intuitive---the lower-cost type $\good$ sets a lower price and the higher-quality type $\good$ receives higher demand. 
Based on Lemma~\ref{lem:D2}, there cannot be two prices on which both types put positive probability and at one of which, demand is positive. 
\begin{lemma} 
\label{lem:D2}
In any equilibrium, if $\sigma_i^*(\good)(P_{\good})>0$ and $\sigma_i^*(\bad)(P_{\bad})>0$, then $D_i(P_{\good})\geq D_i(P_{\bad})$, and if in addition $0<D_i(P_{\bad})\leq D_i(P_{\good})$, then $P_{\good}\leq P_{\bad}$. 
\end{lemma}
The proofs of this and subsequent results are in Appendix~\ref{sec:proofs}. 


The next lemma shows that in any equilibrium satisfying the Intuitive Criterion in which not all consumers buy at price $c_{\bad}$ and belief $\mu_0$, neither firm sets a price at which demand is zero. Both types of both firms make positive profit, and the types set different prices with positive probability. 
To state the lemma, define $v(x)$ as the (unique) consumer valuation $v$ that satisfies $\mu_0h(v)+(1-\mu_0)v =x$, and define 
\begin{align}
\label{mstar}
m^*:=\min_{x\in [c_{\bad}, \mu_0h(\overline{v})+(1-\mu_0)\overline{v}-m]}x\left(\frac{1-F_{v}(h^{-1}(x))}{1-F_{v}(v(x))}-1\right).
\end{align} 
The function $\frac{1-F_{v}(h^{-1}(\cdot))}{1-F_{v}(v(\cdot))}$ 
is the ratio of demand at belief $1$ to demand at belief $\mu_0$. This function only depends on the primitives $F_{v},h,\mu_0$, so $m^*$ only depends on exogenous parameters. The ratio of demands is strictly greater than $1$ and continuous when the denominator is positive (as is the case when $x\leq \mu_0h(\overline{v})+(1-\mu_0)\overline{v}-m$), so $m^*>0$. 


\begin{lemma}
\label{lem:posprofit}
For any $m< m^*$, $i\in\set{\x,\y}$ and $\theta\in\set{\good,\bad}$, in any equilibrium satisfying the Intuitive Criterion, we have
$\pi_{i\theta}^*>0$ and there exists $P_i\geq c_{\bad}^{+}$ s.t.\ $D_i(P_i)>0$ and $\sigma_i^*(\bad)(P_i)>0=\sigma_i^*(\good)(P_i)$. 
\end{lemma}

Lemma~\ref{lem:posprofit} provides the first component of the race to the bottom, namely the good types separating (at least partially) from the bad by setting a lower price. The Intuitive Criterion drives the separation, because it eliminates belief threats at low prices, which would otherwise deter the good types from price-cutting. 

The next lemma establishes a lower bound on the equilibrium price by showing that the good types price weakly above the cost of the bad type. 
\begin{lemma}
\label{lem:Hprice}
For any $m< m^*$ and $i\in\set{\x,\y}$, in any equilibrium satisfying the Intuitive Criterion, if $P_i< c_{\bad}$, then $\sigma_i^*(\good)(P_i)=0$. 
\end{lemma}

The intuition for Lemma~\ref{lem:Hprice} is that the firms' good types are in a \emph{race to the top} at prices in $[0,c_{\bad})$.\footnote{A similar race occurs in  \cite{diamond1971} at all prices below the monopoly level.} Neither firm's good type loses customers to the other firm when raising price slightly, because the small price difference does not motivate the customers to pay the search cost. The reason that a good type does not increase price from $c_{\bad}$ to $c_{\bad}^{+}$ is that the bad type is choosing $c_{\bad}^{+}$, thus belief and demand are significantly lower at $c_{\bad}^{+}$.

In the unique equilibrium surviving the Intuitive Criterion, each good type sets price $c_{\bad}$ and each bad type $c_{\bad}^{+}$, as shown in the following Theorem. 
The proof provides the second component of the race to the bottom: a bad type reduces price to deter its customers from learning and to undercut the other firm's bad type. The motive for the customers to learn comes from the good types separating (the first component of the race, Lemma~\ref{lem:posprofit}), which makes the other firm's price informative and smaller in expectation than a bad type's price. 
\begin{thm}
\label{thm:unique}
For any $m< m^*$ and $i\in\set{\x,\y}$, in the unique equilibrium satisfying the Intuitive Criterion, $\sigma_i^*(\good)(c_{\bad})=1=\sigma_i^*(\bad)(c_{\bad}^{+})$. 
\end{thm}

Theorem~\ref{thm:unique} shows that the unique equilibrium outcome 
that satisfies the Intuitive Criterion is the guessed equilibrium from above. Prices are close to competitive and there is positive, but small price dispersion. The equilibrium outcome is robust to changing the prior, the learning cost, the distribution of consumer valuations and the good type's cost in a range of parameters\footnote{The range defined by $h(v)\geq v$, $h'>1$, $h(\overline{v})<\infty$,  $0<m<\min\{c_{\ell},\;\overline{v}-c_{\bad},\;\frac{\mu_0 c_{\bad}}{1-\mu_0},\;m^*\}$, $c_{\ell}\leq \mu_0 (h(c_{\bad}^{+})-c_{\bad})$, $c_{\bad}=km\geq h(0)-m>0$ for some $k\in\mathbb{N}$, $\frac{d}{dP}P[1-F_v(h^{-1}(P))]>0$ for $P\in[0,c_{\bad}+m]$. 
} 
(Section~\ref{sec:robust} discusses cases outside that range). 
The equilibrium in Theorem~\ref{thm:unique} is distinct from signalling by a monopoly, because a bad type monopolist does not have an incentive to cut price when the good type's price is low enough. This is because there is no competing firm for the customers to learn about and leave to. The bad type sets its monopoly price. Under the Intuitive Criterion, Lemmas~\ref{lem:D2}--\ref{lem:Hprice} still apply, so the good type monopolist sets a price between $c_{\bad}$ and $P_{\good}^{m}$. Separation from the bad type usually requires the good type's price to be strictly below $P_{\good}^{m}$, so unobservable type has some of the same pro-competitive effect with one firm as with two. However, more than one firm is needed for both types' prices to be close to competitive. 

Section~\ref{sec:completeinfo} below contrasts Theorem~\ref{thm:unique} with competition when the type is observed together with the price. The comparisons to monopoly and observed type show that the combination of signalling and multiple firms is necessary as well as sufficient to overcome the effect of the positive search cost. 

Bertrand competition under zero learning cost between two known bad or two known good types leads to equal profits (close to zero) for the firms and no price dispersion, unlike in the equilibrium in Theorem~\ref{thm:unique}. Bertrand competition between a good and a bad firm yields zero demand for the bad firm, but positive demand and profit for the good firm, which sets a strictly higher price than the bad. This differs from the outcome in Theorem~\ref{thm:unique} where a firm that sets a strictly lower price is preferred by the consumers and gets greater demand and profit. 

If some consumers have zero and others positive learning cost, but there is no quality or cost uncertainty, then the firms mix over an interval of prices between the competitive and the monopoly price. The price distribution depends strongly on the density of the learning costs at zero, and whether there is an atom at zero. In the current paper, each firm sets a single price and the equilibrium is robust to perturbing the parameters within a range. 


With consumer taste shocks (horizontal differentiation of firms), there is no price dispersion, and for each firm, some consumers initially at it learn another firm's price and leave. This differs from the current paper, which models vertical differentiation and shows that consumers initially at a good firm do not learn or leave. 

Models of repeat purchases have many equilibria, some of which replicate the pricing patterns found in this paper. However, the markets described by repeated games with high discount factors differ from the markets studied in this paper, which involve infrequent buying (repair services, insurance, durable goods such as cars) and are thus closer to one-shot interactions. 

The next section relaxes some of the assumptions made above. The equilibrium remains qualitatively similar, in particular the Diamond paradox is still resolved.

\section{Robustness}
\label{sec:robust}

Relaxing the assumption that the full-information monopoly price $P_{\good}^{m}$ of the good type is above the cost of the bad type, 
the equilibrium price of the good type is either $c_{\bad}$ as above (if $P_{\good}^{m}= c_{\bad}$), or $P_{\good}^{m}<c_{\bad}$. In the latter case, the only modification of the equilibrium in Section~\ref{sec:existence} is that $\good$ sets price $P_{\good}^{m}\in(0,c_{\bad})$. The proofs simplify, because $\good$ fully separates. 

If the learning cost is large enough ($c_{\ell}> \mu_0[h(c_{\bad}^{+})-c_{\bad}]$), then some customers initially at a bad type setting price $c_{\bad}^{+}$ buy immediately instead of learning the other firm's price. These customers are called \emph{captive}.\footnote{The captive customers correspond to the uninformed customers in \cite{varian1980}.} The mass of captive customers depends on $c_{\ell}- \mu_0[h(c_{\bad}^{+})-c_{\bad}]$. If this is large, then the bad type sets price $P>c_{\bad}^{+}$ with positive probability, because extracting more revenue from captive customers outweighs losing some non-captive ones to the competitor. The probability that the bad type puts on $P>c_{\bad}^{+}$ and the maximal $\hat{P}$ with $\sigma_i(\bad)(\hat{P})>0$ increase in the mass of captive customers. As $c_{\ell}- \mu_0[h(c_{\bad}^{+})-c_{\bad}]$ increases, eventually the good type starts putting positive probability on $c_{\bad}^{+}$. The qualitative features of the model are preserved, in that price is lower than with complete information, and there is price dispersion. 

If there is a distribution of learning costs with $\min c_{\ell}>m$ and $\max c_{\ell}\leq \mu_0[h(c_{\bad}^{+})-c_{\bad}]$, then the equilibrium outcome is unchanged. Learning costs strictly greater than $\mu_0[h(c_{\bad}^{+})-c_{\bad}]$ create captive consumers, as discussed above. 

Nonpositive learning costs for some consumers eliminate the Diamond paradox even without incomplete information, as the previous literature showed. In the current model, enough consumers with a nonpositive learning cost make the good types reduce price, but the positive probability of the other firm having a bad type ensures that the good types do not reach zero price (their marginal cost). The customers initially at a bad type are captive for the other firm's good type. 

If all consumers buy at price $c_{\bad}^{+}$ and belief $\mu_0$ (formally, $h(0)\geq c_{\bad}^{+}/\mu_0$), then there is no reason for a good type to reduce price below $c_{\bad}^{+}$ to increase belief. Both firms pooling on $P_0:=\max\{P\in S_P:P\leq \mu_0h(0)\}$ survives the Intuitive Criterion, because if belief at any $P_1>P_0$ is set to $1$ and the good type wants to deviate to $P_1$, then the bad type also wants to deviate. If the bad, but not the good type wants to deviate to a price, then the Intuitive Criterion sets belief at that price to $0$, which deters deviations. 


The results remain unchanged if the tie-breaking rule for $\sigma_2(v,P_i,P_j)$ in Definition~\ref{def:mix} depends on the belief or the price, e.g.\ if $\mu_{\x}(P_{\x})h(v)+(1-\mu_{\x}(P_{\x}))v-P_{\x} =\mu_{\y}(P_{\y})h(v)+(1-\mu_{\y}(P_{\y}))v-P_{\y}$, then the customer buys from the firm with the greater $\mu_i(P_i)$ (or smaller $P_i$) with probability $p\in[0,1]$. The results also do not change if ties are always broken in favour of a particular firm, say $\x$. 
A slight change in equilibrium is possible if the tie-breaking rule can depend on both the price and the firm, e.g.\ if both firms set $P=c_{\bad}+2m$, then ties are broken in favour of $\x$, but if both set $P=c_{\bad}^{+}$, then in favour of $\y$. In this case, the equilibrium features $\sigma_{\x}(\bad)(c_{\bad}+2m)=1 =\sigma_{\y}(\bad)(c_{\bad}^{+})$. Firm $\y$'s type $\bad$ has no incentive to raise price, because then it would lose all customers. If $\x$ cuts price to $c_{\bad}^{+}$, it still gets zero demand. 
The price at which trade occurs when both firms are of type $\bad$ is still $c_{\bad}^{+}$. Other parts of the equilibrium are unchanged. 

A small asymmetry between firms has a similar effect to asymmetric tie-breaking. Denote type $\theta$ of firm $i$ by $i\theta $. If consumers slightly prefer $\x\bad $ to $\y\bad$, other things equal (interpreted as $\y\bad$ having lower quality), then $\y\bad $ gets zero demand and profit at equal price to $\x\bad $, because consumers at $i\bad$ learn. Then $\x\bad $ sets either the same price as $\y\bad $, or higher by just enough to deter consumers from switching to $\y\bad $. 
Consumers initially at a good type do not learn, unless the quality is lower or price higher than that expected from the other firm's good type, and the difference multiplied by the prior outweighs the learning cost. If consumers do not learn, then they cannot switch firms, so both firms' good types set price $c_{\bad}$, as before. 
Now suppose the firms have the same quality, but the costs satisfy $0=c_{\x\good }\leq c_{\y\good}\leq c_{\x\bad } \leq c_{\y\bad}-m$, and the full-information monopoly price of $\x\theta$ is above $c_{\y\theta}$. Then $\x\bad$ sets price $c_{\y\bad}$, because all consumers at $\x\bad $ learn, so there is asymmetric Bertrand competition between $\x\bad $ and $\y\bad$. Consumers at $\x\good $ do not learn, so the price of $\x\good$ is at least $c_{\y\good}$. The good types are in a race to the top, as in Section~\ref{sec:existence}, so the good types set price $c_{\x\bad}$. 

If the firms can set any price in $[c_{\bad}-\rho,c_{\bad}+\rho]$ for some $\rho>0$ (not constrained to a grid), then 
an equilibrium satisfying the Intuitive Criterion does not exist. The proofs of Lemmas~\ref{lem:D2}--\ref{lem:Hprice} still work, but in Theorem~\ref{thm:unique}, the bad types Bertrand compete down to price $c_{\bad}$. Then belief at $c_{\bad}$ is strictly lower than $1$, the belief at any $P<c_{\bad}$. This makes the payoff of a good type drop discontinuously at $c_{\bad}$, so a best response of a good type does not exist. 
Without refining with the Intuitive Criterion, 
equilibria exist, e.g.\ pooling on $c_{\bad}+\epsilon$ for $\epsilon\in(0,\rho)$ small. This is supported by zero belief for any $P\neq c_{\bad}+\epsilon$. 


Having more than two firms only strengthens competition. Because the bad types do not set the weakly dominated price $c_{\bad}$, and consumers initially at the good types do not learn, pricing cannot get more competitive than with two firms. The outcome is the same as in Section~\ref{sec:existence}. 

More than two types (with higher quality implying lower cost) are conceptually similar to two, but notationally cumbersome. The worst type (highest cost, lowest quality) behaves like $\bad$. In particular, the worst types undercut each other in Bertrand fashion, until they price $m$ above their cost. Consumers initially facing the worst type's price learn the price of the other firm, hoping to meet a better type with a lower price. The Intuitive Criterion imposes (partial) separation of types, so the gain from learning is positive. 
The best type acts similarly to $\good$, setting a price equal to the second-best type's cost. The reason is a race to the top among the best types, as in the baseline model.
Types other than the worst and the best set prices between the second-best and the worst type's cost and may mix, because customers who switch away from the worst type of the other firm are captive for types other than the worst. 

Two-dimensional types with combinations of cost and quality $(c_{\good},\hat{q}_{\good})$, $(c_{\good},\hat{q}_{\bad})$, $(c_{\bad},\hat{q}_{\good})$ and $(c_{\bad},\hat{q}_{\bad})$ are similar to the two-type case when cost and quality are negatively correlated. A type $(c_{\theta},\hat{q}_{\good})$ cannot separate from $(c_{\theta},\hat{q}_{\bad})$ for any $\theta\in\set{\good,\bad}$ in any equilibrium, because $(c_{\theta},\hat{q}_{\bad})$ can imitate any pricing strategy of $(c_{\theta},\hat{q}_{\good})$. The type $(c_{\theta},\hat{q}_{\bad})$ strictly prefers to imitate and get exactly the same payoff as $(c_{\theta},\hat{q}_{\good})$, because demand is based on the quality that the consumers expect, given a price. Demand is thus greater at prices set by $(c_{\theta},q_{\good})$. 
The model with multidimensional types and negative correlation of cost and quality thus reduces to the two-type model in Section~\ref{sec:setup}, with $q_{\theta} =\hat{q}_{\good}\Pr(\hat{q}_{\good}|c_{\theta}) +\hat{q}_{\bad}\Pr(\hat{q}_{\bad}|c_{\theta})$ for $\theta\in\set{\good,\bad}$. 

If the correlation of cost and quality is positive, then the four-type model reduces to the case of two types with higher cost implying higher quality. Price signalling is then directed upward (the high-quality type sets a higher price). The race to the bottom does not occur. Each type sets a price weakly greater than its monopoly price. 

If the correlation of cost and quality is zero, then signalling is impossible in either direction. Consumers expect the average quality after each price set in equilibrium and each type of firm sets its monopoly price given the expected quality. 

Suppose that the firms can advertise as well as signal by price. 
If ads reveal prices to some consumers, then competition increases and the good types cut prices below $c_{\bad}$. The bad types still set price $c_{\bad}^{+}$. If all consumers see both firms' prices, then the good types Bertrand compete to price $m$. 

If ads do not reveal prices, but are just wasteful signalling which for some reason is cheaper for the good type, then the results depend on the noisiness,  timing and cost of the ads. If consumers cannot see the advertising expenditure, but must infer it from noisily observed ad quality and quantity, then ads seen before the prices only change the prior. The results are unaffected by the prior $\mu_0$ if $\mu_0>\max\{\frac{m}{m+c_{\bad}},\; \frac{c_{\ell}}{h(c_{\bad}^{+})-c_{\bad}}\}$. 
Ads seen after the prices have no effect, because the prices already reveal the types. Even if ads are free for the good type, the good type still signals by price, because ads are noisy, so revealing the type via price discretely increases demand. 

Suppose that ads are perfect signals of the money spent on them. Then the relative cost to the types per unit of ads vs per unit of price decrease determines which signalling channel the good type uses. If revealing the type via ads is relatively cheaper, then the good type sets its full-information monopoly price and signals using ads. If the ad costs for the types are similar relative to the difference between the profits lost by cutting price, then ads are not used and the outcome is the equilibrium found above. A similar reasoning applies to any other way to signal, e.g.\ warranties, hiring independent quality testers, etc. 

If each firm trembles when setting price, and prices are the only way to signal, then the results depend on the trembles. 
Denote by $\Pr(P_1|P_2)$ the probability that the consumers see price $P_1$ when the firm tries to set $P_2$. A natural benchmark has $\Pr(P_1|P_2)$ strictly decreasing in $|P_1-P_2|$, and $\Pr(P_1|P_2)>0$ for all $P_1,P_2\in S_P$. 
Reasoning similar to Lemma~\ref{lem:D2} shows that in any equilibrium, the good type tries to set a lower price than the bad. Pooling cannot occur, because then the posterior belief equals the prior at every price, motivating the good type to set a strictly smaller price than the bad. If the trembles are small enough, i.e.\ $\Pr(P|P)\approx 1$ for all $P$, then the distinct prices of the types motivate the consumers to learn. This starts the race to the bottom discussed in Section~\ref{sec:existence}, leading to the same outcome. 

\subsection{Comparison to observable types}
\label{sec:completeinfo}

In this section, the only difference from Section~\ref{sec:setup} is that the type is not inferred from the price, but seen directly. The consumers initially at firm $i$ see the price and type of firm $i$, but have to pay $c_{\ell}$ to learn the price and type of firm $j$. In such a market, prices are not competitive, as shown below.
The equilibrium definition omits part (g) of Definition~\ref{def:mix} and replaces $\mu_i(P_i)$ with $1$ if firm $i$ is of type $\good$ and $0$ if $\bad$. 
The following Proposition puts a lower bound on the price of type $\good$. 
\begin{prop}
\label{prop:complete}
In any equilibrium with observable types, $\pi_{i\good}^*>0$, and if $\sigma_i(\good)(P)>0$, then $P\geq \min\{P_{\good}^m,\; h(c_{\bad}^{+})-m\}$ for $i\in\set{\x,\y}$. 
\end{prop}
The idea for Proposition~\ref{prop:complete} is that race to the top between the good types now continues at prices above $c_{\bad}$, as long as the profit increases in the price and consumers initially at a good type do not learn. If the consumers learn, then with positive probability they switch to the other firm (otherwise there would be no reason to pay the learning cost) and the good type loses demand. The prices of the good types stay close to each other throughout the race to the top, so the motive for a consumer to learn is to find a bad type of the other firm at a price low enough to compensate for the quality difference and the learning cost. So the good types can price above $c_{\bad}^{+}$ by at least 
the quality difference plus the learning cost. 

The race to the top may end at the good type's monopoly price or below that. If the race ends below $P_{\good}^{m}$, then consumers initially at a good type learn and switch with positive probability. The bad type then gets positive demand, even when pricing above the other firm's bad type. The captive customers of the bad type then motivate it to raise price above $c_{\bad}^{+}$. 
%
%
In summary, if quality is seen together with the price, then 
either the good type sets its monopoly price or both types set a price strictly greater than with unobservable types.

\section{Conclusion}
\label{sec:conclusion}

The famous paradox of \cite{diamond1971} is that a market with multiple firms need not be competitive if consumers have to pay a cost to learn the prices of firms. However, as shown in the current paper, negatively correlated production cost and quality that are private information restore competitive pricing. This result holds for a wide range of quality and cost differences between firms. 
There are several mechanisms that make cost and quality negatively correlated across firms, for example economies of scale, regulation or differing managerial talent. These mechanisms operate in many markets. Private information about cost and quality, as well as prices close to the competitive level are empirically reasonable in skilled services, insurance and durable goods, among others. 

The previous literature resolves the Diamond paradox assuming either (a) zero learning cost for a positive fraction of consumers, (b) that consumers observe multiple prices at once, (c) large private taste shocks, or (d) repeat purchases. The current paper models markets in which a given consumer purchases rarely, e.g.\ cars, insurance, repair services, and in which the vertical quality difference is more important than the horizontal taste shock. 
The predictions of the current paper differ from zero search costs and observing multiple prices at once, because the firms set deterministic prices instead of mixing, and the mark-up and profit are larger for a lower-price firm. The current paper assumes no repeat buying of the same good (insurance policies and car models change by the time the consumer purchases a replacement), which distinguishes the model from the literature on repeat purchases. With taste shocks, prices decrease in the number of firms and the degree of product differentiation. 
In the current paper, prices stay constant when the number of firms rises above two or when the quality difference changes within some bounds. 

If lower cost implies higher quality, then a low-cost firm would like to tell consumers about its cost level. A cheap talk message about low cost does not work, for the same reason as cheap talk about high quality has little effect. On the other hand, a low price is a credible signal, because it is differentially costly to the firm types. In some markets, other costly signals are available, e.g.\ warranties or advertising. 
In other applications like insurance, warranties are uncommon, so price signalling is more likely. Even if feasible, signalling by ads or warranties may not be optimal, for example when price signals are cheaper or more precise. 

Signalling by a low price resembles limit pricing, in which an incumbent tries to keep an entrant out of the market. The incumbent sets a low price to convince the entrant that the incumbent has a low cost and is likely to start a price war. The low price in limit pricing is anti-competitive. In the current work, the low price results from competition, thus has different policy implications. A regulator maximising total or consumer surplus should encourage the race to the bottom in prices, for example by punishing low quality or checking the quality of a firm with a larger market share more frequently.


\appendix
\section{Verification of the guessed equilibrium}
\label{sec:existenceproof}

Consumers are clearly best responding to their beliefs in parts 3--5 of the guessed equilibrium. Beliefs in part 1 are consistent with part 2. It remains to check whether firms are best responding in part 2. 
First, downward deviations of type $\good$ are ruled out. The preparative  Lemma~\ref{lem:devprofit} derives the profit function of $\good$ from setting $P\leq c_{\bad}$. 
\begin{lemma}
\label{lem:devprofit}
In the guessed equilibrium, the profit of a type $\good$ firm from $P\leq c_{\bad}$ is 
\begin{align}
\label{Hprofit}
\frac{1}{2}P\left[1-F_v(h^{-1}(P))+(1-\mu_0)\int_{h^{-1}(P)}^{\overline{v}}\sigma_1^*(v,c_{\bad}^{+})(\ell)dF_v(v)\right]. 
\end{align}
\end{lemma}
\begin{proof}
The profit~(\ref{Hprofit}) is derived from~(\ref{demand}) by substituting in the consumers' strategies in the guessed equilibrium: $\sigma_1^*(v,P_i)(b)=1$ and $\sigma_1^*(v,P_i)(\ell)=0$ for consumers initially at $i$, because $P_i\leq c_{\bad}$ and $\mu_i(P_i)=1$. 
Consumers with $v\geq h^{-1}(P_i)$ buy from $i$, and they are a fraction $1-F_v(h^{-1}(P_i))$ of the mass of consumers initially at $i$. 

If firm $j$ is type $\good$, then $\sigma_1^*(v,P_j)(\ell)=0$. 
With probability $1-\mu_0$, firm $j$ is type $\bad$, in which case consumer $v$ at firm $j$ learns $P_i$ with probability $\sigma_1^*(v,c_{\bad}^{+})(\ell)$ and then buys if $v\geq h^{-1}(P_i)$.
\end{proof}
Next, the technical Lemma~\ref{lem:ell} simplifies~(\ref{Hprofit}) by showing that if $\sigma_i^*(\bad)(c_{\bad}^{+})=1$ and $\sigma_i^*(\good)(c_{\bad})=1$ for $i\in\set{\x,\y}$, then $\sigma_1^*(v,c_{\bad}^{+})(\ell)$ is a step function increasing in $v$. 
\begin{lemma}
\label{lem:ell}
For customers initially at a type $\bad$ firm, there exists $v_{01}\in[h^{-1}(c_{\bad}),\overline{v}]$ s.t.\ $\sigma_1^*(v,c_{\bad}^{+})(\ell)=0$ for $v<v_{01}$ and $1$ for $v>v_{01}$. 
\end{lemma}
\begin{proof}
Denote the type $\bad$ firm by $i$. Due to $P_j\leq c_{\bad}^{+}$, in Definition~\ref{def:mix}(d), $v-c_{\bad}^{+}$ may be dropped under the $\max$ w.l.o.g. 
If $\sum_{P_{j}\in S_P}\max\{0,\;\w(v,j,P_{j})\}[\mu_0\sigma_j^*(\good)(P_{j}) +(1-\mu_0)\sigma_j^*(\bad)(P_{j})] -c_{\ell}\geq 0$, then the inequality is strict for all $\hat{v}>v$. 

If $\w(v,j,P_{j})\leq 0$, then $v-c_{\bad}^{+}<0$, so the first inequality in Definition~\ref{def:mix}(d) holds. If $\w(v,j,P_{j})> 0$, then $0$ may be dropped under the $\max$ w.l.o.g. Then $h'>1$ and $\sum_{P_{j}\in S_P}[\mu_0\sigma_j^*(\good)(P_{j}) +(1-\mu_0)\sigma_j^*(\bad)(P_{j})]=1$ imply that the first inequality in Definition~\ref{def:mix}(d) is strict for all $\hat{v}>v_1$. 
So if $\sigma_1^*(v,c_{\bad}^{+})(\ell)>0$, then for all $\hat{v}>v$, $\sigma_1^*(\hat{v},c_{\bad}^{+})(\ell)=1$. 
Taking $v_{01}:=\inf\set{v:\sigma_1^*(v,c_{\bad}^{+})(\ell)>0}$ ensures that $\sigma_1^*(\hat{v},c_{\bad}^{+})(\ell)=0$ for $\hat{v}<v_{01}$ and $1$ for $\hat{v}>v_{01}$.

To prove $v_{01}\geq h^{-1}(c_{\bad})$, note that $h^{-1}(x)<x\;\forall x$, so $h^{-1}(c_{\bad})-c_{\bad}^{+}<0$. If $P_j\geq c_{\bad}$, then $\w(h^{-1}(c_{\bad}),j,P_{j})\leq 0$. The $-c_{\ell}$ term in Definition~\ref{def:mix}(d) then ensures $\sigma_1^*(h^{-1}(c_{\bad}),c_{\bad}^{+})(\ell)=0$. 
\end{proof}

Downward deviations by a type $\good$ firm are ruled out in the following Lemma. 
After that, the incentives of firm type $\bad$ are discussed, and then the deviations of $\good$ to $P_{\good}>c_{\bad}$ are ruled out. 
\begin{lemma}
\label{lem:Pgreatercb}
A type $\good$ firm's best response to the strategies of other players in the guessed equilibrium satisfies $P\geq c_{\bad}$. 
\end{lemma}
\begin{proof}
Based on Lemma~\ref{lem:ell}, $\sigma_1^*(v,c_{\bad}^{+})(\ell)=0$ for all $v\leq h^{-1}(c_{\bad})\geq h^{-1}(P)$, where $P\leq c_{\bad}$. Therefore~(\ref{Hprofit}) reduces to $\frac{1}{2}P[1-F_{v}(h^{-1}(P))+(1-\mu_0)[1-F_{v}(v_{01})]]$, with $v_{01}$ independent of $P$. 
The assumption $P_{\good}^{m} :=\arg\max_P P[1-F_{v}(h^{-1}(P))] \geq c_{\bad}^{+}$ then implies $\arg\max_P \frac{1}{2}P[1-F_{v}(h^{-1}(P))+(1-\mu_0)[1-F_{v}(v_{01})]]\geq c_{\bad}^{+}$, because if $P_{\good}^{m}D(P_{\good}^{m})\geq PD(P)$ for all $P\leq P_{\good}^{m}$, then for any $\bar{D}>0$ and $P\leq P_{\good}^{m}$, we have $P_{\good}^{m}D(P_{\good}^{m})+P_{\good}^{m}\bar{D}\geq PD(P)+P\bar{D}$. 
So type $\good$ optimally sets a price $P\geq c_{\bad}$. 
\end{proof}

\begin{lemma}
\label{lem:Bprice}
In the guessed equilibrium, a type $\bad$ firm's best response to the strategies of other players is $P=c_{\bad}^{+}$.
\end{lemma}
\begin{proof}
A type $\bad$ firm clearly does not deviate to $P\leq c_{\bad}$, which is weakly dominated by $P=c_{\bad}^{+}$. 
Consider $\bad$'s deviations to $P> c_{\bad}^{+}$.
Parts 1 and 4 of the guessed equilibrium ensure that each customer initially at a firm charging $P\geq c_{\bad}^{+}$ either does not buy or learns the prices of both firms. 
If the other firm is type $\good$, then by part 5 of the guessed equilibrium, all customers leave a $\bad$ firm. Therefore $\bad$ sets a price that best responds to facing $\bad$ with certainty. If both firms are type $\bad$, then Bertrand competition leads both to set price $c_{\bad}^{+}$. 
\end{proof} 

Having ruled out deviations by $\bad$, the final step (Lemma~\ref{lem:Plesscb}) is to eliminate upward deviations by a type $\good$ firm. 
\begin{lemma}
\label{lem:Plesscb}
A type $\good$ firm's best response to the strategies of other players in the guessed equilibrium satisfies $P\leq c_{\bad}$. 
\end{lemma}
\begin{proof}
If a type $\good$ firm $i$ sets $P>c_{\bad}^{+}$, then it gets zero demand in the guessed equilibrium, because $\mu_i(P)=0$ and the other firm $j$ is expected to set price $P_j\leq c_{\bad}^{+}<P$. So $P>c_{\bad}^{+}$ is not a profitable deviation. 
At $P=c_{\bad}^{+}$, belief is $\mu_i(P)=0$ and demand is $\frac{1-\mu_0}{2}[1-F_{v}(c_{\bad}^{+})]$. The profit of $\good$ is then $c_{\bad}^{+}\frac{1-\mu_0}{2}[1-F_{v}(c_{\bad}^{+})]$, which is less than the profit at $P=c_{\bad}$ iff $c_{\bad}^{+}(1-\mu_0)[1-F_{v}(c_{\bad}^{+})]\leq c_{\bad}[1-F_{v}(c_{\bad})+(1-\mu_0)[1-F_{v}(v_{01})]]$. Sufficient for this is 
$m(1-\mu_0) \leq \mu_0c_{\bad}$, which holds by assumption. 
\end{proof}
Combining Lemmas~\ref{lem:devprofit}--\ref{lem:Plesscb},
the guessed equilibrium is verified.

\section{Proofs omitted from the main text}
\label{sec:proofs}

\begin{proof}[Proof of Lemma~\ref{lem:D2}]
In any equilibrium, the incentive constraints (ICs) $P_{\good} D_i(P_{\good})\geq P D_i(P)$ and $(P_{\bad}-c_{\bad}) D_i(P_{\bad})\geq (P-c_{\bad}) D_i(P)$ hold for any $P,P_{\good},P_{\bad}$ s.t.\ $\sigma_i^*(\theta)(P_{\theta})>0$. 
Demand and price are nonnegative and finite by definition.
From $(P_{\bad}-c_{\bad}) D_i(P_{\bad})\geq (P_{\good}-c_{\bad}) D_i(P_{\good})$ and $P_{\good} D_i(P_{\good})\geq P_{\bad} D_i(P_{\bad})$, we get $(P_{\bad}-c_{\bad}) D_i(P_{\bad})\geq P_{\good} D_i(P_{\good})-c_{\bad}D_i(P_{\good})\geq P_{\bad} D_i(P_{\bad})-c_{\bad}D_i(P_{\good})$, so $D_i(P_{\bad})\leq D_i(P_{\good})$. 

If $0<D_i(P_{\bad})\leq D_i(P_{\good})$ and $(P_{\bad}-c_{\bad}) D_i(P_{\bad})\geq (P_{\good}-c_{\bad}) D_i(P_{\good})$, then $P_{\bad}-c_{\bad}\geq P_{\good}-c_{\bad}$, so $P_{\bad}\geq P_{\good}$. 
\end{proof}

\begin{proof}[Proof of Lemma~\ref{lem:posprofit}]
Suppose that for both $i\in\set{\x,\y}$ there exists $P_{i0}$ s.t.\ $D_i(P_{i0})=0$ and $\sigma_i^*(\bad)( P_{i0})>0$. 
Then $\pi_{i\bad}^*=0$, otherwise $\bad$ would deviate to put probability $1$ on prices at which profit is positive. 
If $P_i<\overline{v}$ on or off the equilibrium path, then for all $v>P_i$, $\w(v,i,P_{i})>0$ and by Definition~\ref{def:mix}, $\sigma_1^*(v,P_i)(n) =0 =\sigma_2^*(v,P_i,P_j)(n_{\ell}) =\sigma_2^*(v,P_j,P_i)(n_{\ell})$ for any $P_j$ and $i\neq j\in\set{\x,\y}$. Then for all $v\in(P_i,\overline{v}]$, $\sigma_1^*(v,P_i)(\ell)+\sigma_1^*(v,P_i)(b)=1$ and $\sigma_2^*(v,P_i,P_j)(b_i) +\sigma_2^*(v,P_i,P_j)(b_j)=1$, i.e.\ the consumer types $v\in(P_i,\overline{v}]$ buy from some firm. Thus $D_i(P_i)+D_j(P_j)\geq 1-F_{v}(P_i)>0$ for any $P_j$ and any $P_i<\overline{v}$. 
With probability at least $\sigma_j^*(\bad)( P_{j0})>0$, 
$D_j(P_j)=0$ and $D_i(P_i)>0$ for any $P_i<\overline{v}$. Thus deviating from $P_{i0}$ to $P\in(c_{\bad},\overline{v})$ yields positive expected profit to $i$, contradicting $\sigma_i^*(\bad)( P_{i0})>0$. Therefore in any equilibrium, type $\bad$ of firm $i$ only sets prices at which profit is positive. 

Because $\good$ can imitate $\sigma^*(\bad)$ at a strictly lower cost, we have $\pi_{i \good}^*\geq\pi_{i \bad}^*$ in any equilibrium, with strict inequality if $\pi_{i \good}^*>0$. 

Next, the Intuitive Criterion is used to show that type $\bad$ of firm $i$ partially separates. 
Suppose that 
if 
$\sigma_i^*(\bad)(P_i)>0$, then  $\sigma_i^*(\good)(P_i)>0$ (no separation of $\bad$). Lemma~\ref{lem:D2} implies that there is at most one $P$ s.t.\ $\sigma_i^*(\theta)(P)>0$ for $\theta\in\set{\good,\bad}$. We supposed that there is no $\hat{P}$ s.t.\ $\sigma_i^*(\bad)(\hat{P})>0 =\sigma_i^*(\good)(\hat{P})$, so $P_i$ is the unique price with $\sigma_i^*(\bad)(P_i)>0$. Thus $\sigma_i^*(\bad)(P_i)=1$, which implies $\mu_i(P_i)\leq \mu_0$. 
From $D_i(P_i)>0$, we get 
$\w(\overline{v},i,P_{i})>0$. 
Define $D_i^{\mu}(P)$ as demand for firm $i$ when belief is fixed at $\mu_i(P)=\mu\in[0,1]$. 
Clearly $D_i^{\mu}(P)$ decreases in $P$ and increases in $\mu$. 
Define $\mathcal{P}_{\neg \bad}:=\{P\in S_P: (P-c_{\bad})D_i^0(P)<\pi_{i\bad}^*>(P-c_{\bad})D_i^1(P)\}$, which is nonempty, because if $P\leq c_{\bad}$, then $ P\in\mathcal{P}_{\neg \bad}$. If $\sigma_i^*(\bad)(\tilde{P})>0$, then $\tilde{P}\notin \mathcal{P}_{\neg \bad}$. 

To apply the Intuitive Criterion, set $\mu_i(P)=1$ for all $P\in \mathcal{P}_{\neg \bad}$. Check whether $\good$ wants to deviate to $P_{d}:=\max\{P\in \mathcal{P}_{\neg \bad}:P< P_i\}\geq c_{\bad}$. 
By definition of $\mathcal{P}_{\neg \bad}$ and $S_P$, there exists $P_n\in(P_d,P_i]$ (possibly equal to $P_i$) s.t.\ $P_d-P_n\leq m$ and $(P_n-c_{\bad})D_i(P_n)\geq \pi_{i\bad}^*=(P_i-c_{\bad})D_i(P_i)$. Focus on the minimal such $P_n$. The definition of $P_d$ implies $P_n\geq c_{\bad}^{+}$. Due to $\pi_{i\good}^* =P_iD_i(P_i) =\pi_{i\bad}^*+c_{\bad}D_i(P_i)$, we have
$P_nD_i(P_n)\geq \pi_{i\bad}^*+c_{\bad}D_i(P_n)
=\pi_{i\good}^* +c_{\bad}D_i(P_n)-c_{\bad}D_i(P_i)$. 

Suppose $P_n<P_i$ (implying $P_d<P_i-m$). Then due to $(P_n-c_{\bad})D_i(P_n)\geq (P_i-c_{\bad})D_i(P_i)=\pi_{i\bad}^*$, we have $D_i(P_n)>D_i(P_i)$ and $P_nD_i(P_n)\geq\pi_{i\good}^*-c_{\bad}D_i(P_i)+c_{\bad}D_i(P_n)>\pi_{i\good}^*$, contradicting equilibrium. The remaining possibility is $P_n=P_i$ (so $P_d=P_i-m$). 

If $h(0)\leq c_{\bad}^{+}$ (which holds by assumption) and $\mu_i(P_i)\leq \mu_0$, then a positive mass of consumers initially at $i$ do not buy from $i$ at $P_i\in[ c_{\bad}^{+},h(\overline{v})]$ (consumer $v=0$ strictly prefers not to). 
In that case, $i$ reducing price to $P_d$ and increasing belief to $\mu_i(P_d)=1> \mu_0\geq \mu_i(P_i)$ strictly increases total demand $D_i(\cdot)+D_j(\cdot)$ by at least $\frac{1}{2}[1-F_{v}(h^{-1}(P_d))-1+F_{v}(v_2)]>0$, where 
$v_2$ solves 
$\w(v_2,i,P_{i})=0$. 
Demand weakly increases as $P_i$ falls, so $D_i^1(P_d)-D_i(P_i)\geq \frac{1}{2}[F_{v}(v_2)-F_{v}(h^{-1}(P_d))]>0$. 

If $m<m^*\leq \frac{P_d[D_i^1(P_d)-D_i(P_i)]}{D_i(P_i)}$ and $P_i-P_d=m$, then $P_dD_i^1(P_d)-P_dD_i(P_i)-m D_i(P_i)>0$, i.e.\ $P_dD_i^1(P_d)>P_iD_i(P_i)=\pi_{i\good}^*$, contradicting equilibrium in the case $P_n=P_i$. 
The Intuitive Criterion thus eliminates equilibria in which $\sigma_i^*(\bad)(P_i)>0\Rightarrow \sigma_i^*(\good)(P_i)>0$, pooling among them. 

Because there exists $P_i\geq c_{\bad}^{+}$ s.t.\ $D_i(P_i)>0$ and $\sigma_i^*(\bad)(P_i)>0=\sigma_i^*(\good)(P_i)$, we have $\mu_i(P_i)=0$ with probability at least $\sigma_i^*(\bad)(P_i)>0$. If either type of firm $j$ sets price $P_j=c_{\bad}^{+}$, then $\mu_j(P_j)\geq \mu_i(P_i)$ and $P_j\leq P_i$ with the positive probability $\sigma_i^*(\bad)(P_i)>0$, so firm $j$ makes positive profit. 

Due to $\pi_{j\theta}^*>0$, we can apply the Intuitive Criterion reasoning above to firm $j$ to prove that type $\bad$ of firm $j$ also separates at least partially. 
\end{proof}

\begin{proof}[Proof of Lemma~\ref{lem:Hprice}]
If $m\leq m^*$, then by Lemma~\ref{lem:posprofit}, $\pi_{i\theta}^*>0$ for both $i\in\set{\x,\y}$ and $\theta\in\set{\good,\bad}$. The Intuitive Criterion then implies $\mu_i(P)=1$ for all $P\leq c_{\bad}$. 

Suppose that $P_{i\good}:=\min\{P:\sigma_i(\good)(P)>0\}\leq \min\{c_{\bad}-m,\;\min\{P:\sigma_j(\good)(P)>0\}+c_{\ell}\}$. If firm $i$ raises price to $P_{i\good}+m$, then due to $m<c_{\ell}$, consumers initially at firm $i$ still choose $\sigma_1(v,P_{i\good}+m)(\ell)=0$. 
The customers at $j$ who chose $\ell$ anticipating $\sigma_i^*(\good)$ do not know about the deviation, so still choose $\ell$. 
Upon learning $P_{i\good}+m$, a customer initially at $j$'s type $\bad$ has a choice between $\bad$ at $P_{\bad}\geq c_{\bad}^{+}$ and $\good$ at $P_{i\good}+m\leq c_{\bad}^{+}$, so still buys from $i$. 
A customer initially at $j$'s type $\good$ still buys from $i$ if $P_{j\good}\geq P_{i\good}+2m$. If $P_{j\good}\in[P_{i\good}, P_{i\good}+m]$, then due to $P_{i\good}\leq c_{\bad}-m$ and $P_{\bad}\geq c_{\bad}^{+}$, the type of $j$ is revealed as $\good$. 

A type $v$ customer initially at $j$ chooses $\sigma_1^*(v,P)(\ell)=1$ only if $h(v)-P_{j\good}\leq \mu_0\sum_{P_{i}\in S_P}\max\{0,\;h(v)-P_{i},\; h(v)-P_{j\good}\}\sigma_i^*(\good)(P_{i}) +(1-\mu_0)\sum_{P_{i}\in S_P}\max\{0,\;v-P_{i},\; h(v)-P_{j\good}\}\sigma_i^*(\bad)(P_{i}) -c_{\ell}\geq 0$. From $P_{i\good}\leq P_{j\good} \leq P_{i\good}+m < c_{\bad}^{+}$, we get $\sum_{P_{i}\in S_P}\max\{0,\;h(v)-P_{i},\; h(v)-P_{j\good}\}\sigma_i^*(\good)(P_{i})\leq \max\{0,\;h(v)-P_{j\good}+m\}$ and $\sum_{P_{i}\in S_P}\max\{0,\;v-P_{i},\; h(v)-P_{j}\}\sigma_i^*(\bad)(P_{i})\leq \max\{0,\;h(v)-P_{j}\}$. Due to $m<c_{\ell}$, if $P_{j\good}\in[P_{i\good}, P_{i\good}+m]$, then $\sigma_j(v,P_{j\good})(\ell)=0$. The customers who might switch away from $i$ after a price increase do not learn both prices, so have no choice of switching.

On $P\leq c_{\bad}$, the profit of $\good$ is then given by~(\ref{Hprofit}) and Lemmas~\ref{lem:ell}--\ref{lem:Pgreatercb} prove that $\sigma_i^*(\good)(P_i)=0$ for all $P_i< c_{\bad}$. 
\end{proof}

\begin{proof}[Proof of Theorem~\ref{thm:unique}]
By Lemma~\ref{lem:posprofit}, 
there exist $P_i,P_j\geq c_{\bad}^{+}$ s.t.\ $D_i(P_i)>0$, $D_j(P_j)>0$, $\sigma_i^*(\bad)(P_i)>0=\sigma_i^*(\good)(P_i)$ and $\sigma_j^*(\bad)(P_j)>0=\sigma_j^*(\good)(P_j)$. Choose the maximal such $P_i,P_j$ and assume $P_i\geq P_j$ w.l.o.g. Clearly $\mu_i(P_i)=0=\mu_j(P_j)$. By Lemma~\ref{lem:D2}, if $\sigma_j^*(\theta)(P_{j\theta})>0$, then $P_{j\good}\leq P_{j\bad}$, so if $\sigma_j^*(\bad)(P_j)>0=\sigma_j^*(\good)(P_j)$, then $P_{j\good}<P_{j}\leq P_i$. 

By Definition~\ref{def:mix}(d),(e), consumer $v$ initially at firm $i$ charging $P_i$ chooses $\sigma_1^*(v,P_i)(b)=0$ if $v-P_{i}\leq \sum_{P\in S_P}\max\{0,\;\w(v,j,P),\;v-P_{i}\}[\mu_0\sigma_j^*(\good)(P) +(1-\mu_0)\sigma_j^*(\bad)(P)] -c_{\ell}$. W.l.o.g.\ $v-P_i$ may be dropped under the $\max$, due to $P\leq P_i$. Sufficient for $\sigma_1^*(v,P_i)(b)=0$ is then 
$v-P_{i}+c_{\ell}\leq \sum_{P\in S_P}[\mu_j(P)h(v)+(1-\mu_j(P))v-P_{i}][\mu_0\sigma_j^*(\good)(P) +(1-\mu_0)\sigma_j^*(\bad)(P)]$,
equivalently $c_{\ell}\leq \sum_{P\in S_P}\mu_j(P)[h(v)-v][\mu_0\sigma_j^*(\good)(P) +(1-\mu_0)\sigma_j^*(\bad)(P)]$. This holds iff $c_{\ell}\leq \mu_0[h(v)-v]$, due to Bayes' rule~(\ref{mu}) 
and $\sum_{P\in S_P}\sigma_j^*(\theta)(P)=1$. 

Consumers $v<P_{i}$ always choose $\sigma_1^*(v,P_i)(b)=0$. Due to $P_{i}\geq c_{\bad}^{+}$, all consumers initially at firm $i$ choose $\sigma_1^*(v,P_i)(b)=0$ if $c_{\ell}\leq \mu_0[h(c_{\bad}^{+})-c_{\bad}^{+}+m]$, i.e.\ $c_{\ell}\leq \mu_0[h(c_{\bad}^{+})-c_{\bad}]$, which holds by assumption. If all consumers at type $\bad$ of firm $i$ choose $\sigma_1^*(v,P_i)(b)=0$, then $\bad$ is in Bertrand competition with firm $j$'s type $\bad$ over the consumers choosing $\sigma_1^*(v,P_i)(\ell)=1$. There is a positive mass of such consumers, because $D_i(P_i)>0$ by Lemma~\ref{lem:posprofit}. If $P_i>P_j$, then $\sigma_2^*(v,P_i,P_j)(b_i)=0=\sigma_2^*(v,P_j,P_i)(b_i)$, 
contradicting $D_i(P_i)>0$. 

If $P_i=P_j>c_{\bad}^{+}$, then deviating from $P_i$ to the adjacent price $P_d=P_i-m$ undercuts firm $j$'s price $P_j$, increasing $D_i$ by at least $D_j(P_j)>0$. If $(P_i-c_{\bad})D_i(P_i) <(P_d-c_{\bad})[D_i(P_d)+D_j(P_j)]$, then the deviation is profitable. Because $P_i-c_{\bad}=P_d+m-c_{\bad}\geq 2m$, we get $P_i-c_{\bad}\leq 2(P_d-c_{\bad})$. Due to $P_i=P_j$, we can focus on either firm, so assume $D_i(P_i)\leq D_j(P_j)$ w.l.o.g. This with $P_i-c_{\bad}\leq 2(P_d-c_{\bad})$ yields $(P_i-c_{\bad})D_i(P_i) \leq(P_d-c_{\bad})[D_i(P_d)+D_j(P_j)]$. Not all consumers buy at $P_i>0$, $\mu_i(P_i)=0$, so $P_d<P_i$ and $\mu_i(P_d)\geq \mu_i(P_i)$ result in $D_i(P_d)>D_i(P_i)$. Undercutting is profitable at $P_i=P_j>c_{\bad}^{+}$, so if $\sigma_i^*(\bad)(P)>0$, then $P\leq c_{\bad}^{+}$. 

If $\sigma_i^*(\bad)(c_{\bad}^{+})>0=\sigma_i^*(\good)(c_{\bad}^{+})$, then $\sigma_i^*(\good)(c_{\bad})=1$ by Lemma~\ref{lem:Hprice}. Finally, $\pi_{i\bad}^*>0$ implies $\sigma_i^*(\bad)(c_{\bad})=0$, so $\sigma_i^*(\bad)(c_{\bad}^{+})=1$. 
\end{proof}

\begin{proof}[Proof of Proposition~\ref{prop:complete}]
Price $m$ is available to type $\good$, with $D_i(m)>0$ regardless of $\sigma_j$, because $m<c_{\ell}$. Therefore $\pi_{i\good}^*>0$ for $i\in\set{\x,\y}$. 

Denote $\min\{P:\sigma_i^*(\good)(P)>0\}$ by $P_{i\good}$ for $i\in\set{\x,\y}$ and assume w.l.o.g.\ $P_{i\good}\leq P_{j\good}<h(\overline{v})$. 
A customer type $v\geq h^{-1}(P_{i\good}+m)$ initially at firm $i$ who sees that the firm is type $\good$ and charges $P_{i\good}+m$ chooses $\sigma_1^*(v,P_{i\good}+m)(\ell)=0$ if 
$h(v)-P_{i\good}-m > V :=\mu_0\sum_{P_{j}\in S_P}\max\{h(v)-P_{j},\; h(v)-P_{i\good}-m\}\sigma_j^*(\good)(P_{j}) +(1-\mu_0)\sum_{P_{j}\in S_P}\max\{v-P_{j},\; h(v)-P_{i\good}-m\}\sigma_j^*(\bad)(P_{j}) -c_{\ell}$. Type $\bad$ sets $P\geq c_{\bad}^{+}$ (which weakly dominates $P\leq c_{\bad}$), firm $j$'s type $\good$ sets $P_{j\good}\geq P_{i\good}$ by assumption, and $\sum_{P_{j}\in S_P}\sigma_j^*(\theta)(P_{j})=1$, so 
$V\leq \mu_0[h(v)-P_{i\good}] +(1-\mu_0)\max\{v-c_{\bad}^{+},\; h(v)-P_{i\good}-m\} -c_{\ell}$. 
Sufficient for customer type $v\geq h^{-1}(P_{i\good}+m)$ facing $P_{i\good}+m$ not to learn is 
$(1-\mu_0)[h(v)-P_{i\good}-m] >\mu_0 m -c_{\ell} +(1-\mu_0)\max\{v-c_{\bad}^{+},\; h(v)-P_{i\good}-m\}$, which holds if 
$P_{i\good}+m <h\left(c_{\bad}^{+} +\frac{c_{\ell}-\mu_0 m }{1-\mu_0}\right)$
So type $\good$ of firm $i$ increases price to at least $\min\{P_{\good}^m,\; h(c_{\bad}^{+})-m\}$. 

Firm $i$ was arbitrary, so the same reasoning applies to firm $j$. 
\end{proof}

\bibliographystyle{ecta}
\bibliography{teooriaPaberid} 
\end{document}